\documentclass[12pt]{amsart}
\usepackage{etex}
\usepackage{amsfonts, amssymb, latexsym}

\usepackage{amscd,amssymb}
\usepackage{verbatim}
\usepackage{pstricks}
\usepackage{pst-node}
\usepackage[mathscr]{euscript}
\usepackage{tikz}
\usetikzlibrary{matrix,arrows}
\usepackage[normalem]{ulem}
\usepackage{varwidth}
\usepackage{leftidx}

\newsymbol\pp 1275
\usepackage{tikz}
\usetikzlibrary{matrix,arrows,backgrounds,shapes.misc,shapes.geometric,patterns,calc,positioning}

\usepackage[colorlinks=true, pdfstartview=FitV, linkcolor=blue, citecolor=blue, urlcolor=blue]{hyperref}

\usepackage{fullpage}
\usepackage{graphicx}
\usepackage{enumerate}

\def\VR{\kern-\arraycolsep\strut\vrule &\kern-\arraycolsep}
\def\vr{\kern-\arraycolsep & \kern-\arraycolsep}

\newtheorem{theorem}{Theorem}

\newtheorem{lemma}[theorem]{Lemma}

\newtheorem{corollary}[theorem]{Corollary}

\theoremstyle{definition}
\newtheorem{definition}[theorem]{Definition}

\newtheorem{rmk}[theorem]{Remark}
\newenvironment{remark}[1][]{\begin{rmk}[#1]\pushQED{\qed}}{\popQED \end{rmk}}

\newtheorem{qu}[theorem]{Question}

\newtheorem*{rmknonum}{Remark}

\newtheorem{obs}[theorem]{Observation}

\newtheorem{ex}[theorem]{Example}

\newcommand{\Hom}{\operatorname{Hom}}
\newcommand{\End}{\operatorname{End}}

\newcommand{\tr}{\operatorname{Tr}}

\newcommand{\rep}{\operatorname{rep}}

\newcommand{\mm}{\operatorname{M}}

\newcommand{\GL}{\operatorname{GL}}

\newcommand{\ZZ}{\mathbb Z}
\newcommand{\E}{\mathbb E}

\newcommand{\RR}{\mathbb R}
\newcommand{\NN}{\mathbb N}

\newcommand{\tup}{\mathbf p}
\newcommand{\0}{\mathbf 0}

\newcommand{\B}{\mathbf B}
\newcommand{\C}{\mathbf C}

\newcommand{\V}{V}

\newcommand{\Id}{\mathbf{I}}

\newcommand{\Mat}{\operatorname{Mat}}

\newcommand{\Stab}{\operatorname{Stab}}

\newcommand{\ddim}{\operatorname{\mathbf{dim}}}
\newcommand{\cc}{\operatorname{\mathbf{c}}}
\newcommand{\dd}{\operatorname{\mathbf{d}}}

\newcommand{\nn}{\mathbf{n}}

\newcommand{\A}{\mathbf A}
\newcommand{\G}{\mathcal{G}}
\newcommand{\Q}{\mathcal{Q}}

\newcommand{\p}{\mathcal{P}}
\newcommand{\ar}{\mathcal{A}}
\newcommand{\s}{\mathcal{S}}
\newcommand{\jj}{\mathcal{I}^{-}_j}
\newcommand{\ii}{\mathcal{I}^{+}_i}

\newcommand{\capa}{\mathbf{cap}}
\newcommand{\Det}{\mathsf{Det}}

\newcommand\restr[2]{{
  \left.\kern-\nulldelimiterspace 
  #1 
  \vphantom{\big|} 
  \right|_{#2} 
  }}

\begin{document}
\title{The capacity of quiver representations and the Anantharam-Jog-Nair inequality}
\author{Calin Chindris and Harm Derksen}
\address{University of Missouri-Columbia, Mathematics Department, Columbia, MO, USA}
\email[Calin Chindris]{chindrisc@missouri.edu}

\address{Northeastern University, Boston, MA}
\email[Harm Derksen]{ha.derksen@northeastern.edu}

\date{\today}
\bibliographystyle{amsalpha}

\begin{abstract} The Anantharam-Jog-Nair inequality \cite{AJN-2022} in Information Theory provides a unifying approach to the information-theoretic form of the Brascamp-Lieb inequality \cite{CarCor-2009} and the Entropy Power inequality \cite{ZamFed-1993, Rio-2011}. 

In this paper, we use methods from Quiver Invariant Theory \cite{ChiDer-2021} to study Anantharam-Jog-Nair inequalities with integral exponents. For such an inequality, we first view its input datum as a quiver datum and show that the best constant that occurs in the Anantharam-Jog-Nair inequality is of the form $-{1\over 2}\log (\capa(V,\sigma))$ where $\capa(V, \sigma)$ is the capacity of a quiver datum $(V, \sigma)$ of a complete bipartite quiver.

The general tools developed in \cite{ChiDer-2021}, when applied to complete bipartite quivers, yield necessary and sufficient conditions for: $(1)$ the finiteness of the Anantharam-Jog-Nair best constants; and $(2)$ the existence of Gaussian extremizers. These results recover some of the main results in \cite{AJN-2022} and \cite{AraCouZha-2022}.  In addition, we characterize gaussian-extremizable data in terms of semi-simple data, and provide a general character formula for the Anatharam-Jog-Nair constants.  Furthermore, our quiver invariant theoretic methods lead to necessary and sufficient conditions for the uniqueness of Gaussian extremizers. This answers the third and last question left unanswered in \cite{AJN-2022}.
\end{abstract}

\maketitle
\setcounter{tocdepth}{1}
\tableofcontents

\section{Introduction} 
\subsection{Motivation} Let $k$ and $m$ be positive integers and let $\dd=(d_1, \ldots, d_k)$ and $\nn=(n_1, \ldots, n_m)$ be two tuples of positive integers with $D:=\sum_{i=1}^k d_i$.

Let $\A:=(A_{ij}: \RR^{d_i} \to \RR^{n_j})_{i \in [k], j\in [m]}$ be a tuple of linear maps such that the maps $A_j:=[A_{1j} \ldots A_{kj}]: \RR^D \to \RR^{n_j}$ are surjective for all $j \in [m]$. 

Let $\cc=(c_1, \ldots c_k)$ and $\tup=(p_1, \ldots, p_m)$ be tuples of positive integers such that 
$$
\sum_{i=1}^k c_i d_i=\sum_{j=1}^m p_j n_j.
$$

Let us define $\p(\dd)$ to be the set of all random vectors in $\RR^D$ that can be partitioned into $k$ components $X=(X_1, \ldots, X_k)$ such that
\begin{enumerate}[(a)]
\item $X_i$ is random vector in $\RR^{d_i}$ with finite entropy and second moment for every $i \in [k]$;
\smallskip

\item $X_1, \ldots, X_k$ are mutually independent;
\smallskip

\item $\E(X)=0$ and $\E||X||^2_2<\infty$.
\end{enumerate}

We also define $\p_g(\dd) \subseteq \p(\dd)$ to be the set of all $X=(X_1, \ldots, X_k) \in \p(\dd)$ such that each $X_i$ is Gaussian. In \cite[Theorem 3]{AJN-2022}, Anantharam, Jog, and Nair have recently proved that the best (\emph{i.e.}, the smallest) constant $C$ for which the entropic inequality 
\begin{equation}\label{AJN-ineq}
\sum_{i=1}^k c_i h(X_i)\leq \sum_{j=1}^m p_j h(A_jX)+C
\end{equation}
holds for all  $X\in \p(\dd)$ can be computed by considering only independent Gaussian variables $(X_1, \ldots, X_k) \in \p_g(\dd)$. Thus the best constant in $(\ref{AJN-ineq})$, denoted by $\mm(\A, \cc, \tup)$, can be expressed as
\begin{equation}\label{AJN-const}
\mm(\A, \cc, \tup)=\sup_{Z \in \p_g(\dd)} \left( \sum_{i=1}^k r_i h(Z_i)-\sum_{j=1}^m p_j h(A_jZ) \right).
\end{equation}
\noindent
We call $(\A, \cc, \tup)$ an AJN-datum and $\cc$ and $\tup$ integral exponents. We also refer to $(\ref{AJN-ineq})$ as an \emph{AJN-inequality} and call $(\ref{AJN-const})$ an \emph{AJN-constant}. In \cite{AJN-2022} (see also \cite{Cou-2019}), the authors have settled the question of when $\mm(\A, \cc, \tup)$ is finite.  Another important question posed in \cite{AJN-2022} asks to determine when extermizers and Gaussian extremizers exist, and whether extremizability is equivalent to Gaussian extremizability. These questions have been answered in \cite{AraCouZha-2022}. Our goal in this paper is to explain how these results, except for \emph{extremizability is the same as Gaussian extremizability}, can be deduced from the quiver invariant theoretic methods developed in \cite{ChiDer-2021}. Moreover, we will see that the general results from \cite{ChiDer-2021} applied to the set-up of the ANJ inequality $(\ref{AJN-ineq})$ yield necessary and sufficient conditions for the uniqueness of Gaussian extremizers. This solves the last question in \cite[Question 3, page 7]{AJN-2022}.

\subsection{Our results} The tuple of linear maps $\A:=(A_{ij}: \RR^{d_i} \to \RR^{n_j})_{i \in [k], j\in [m]}$ and the exponent $\cc$ that appear in $(\ref{AJN-ineq})$ and $(\ref{AJN-const})$ can be viewed as an arrangement of vectors spaces and linear maps attached to the vertices and arrows of the complete bipartite directed graph $\Q_{k,m}$ as follows

$$\Q_{n,m}:~
\vcenter
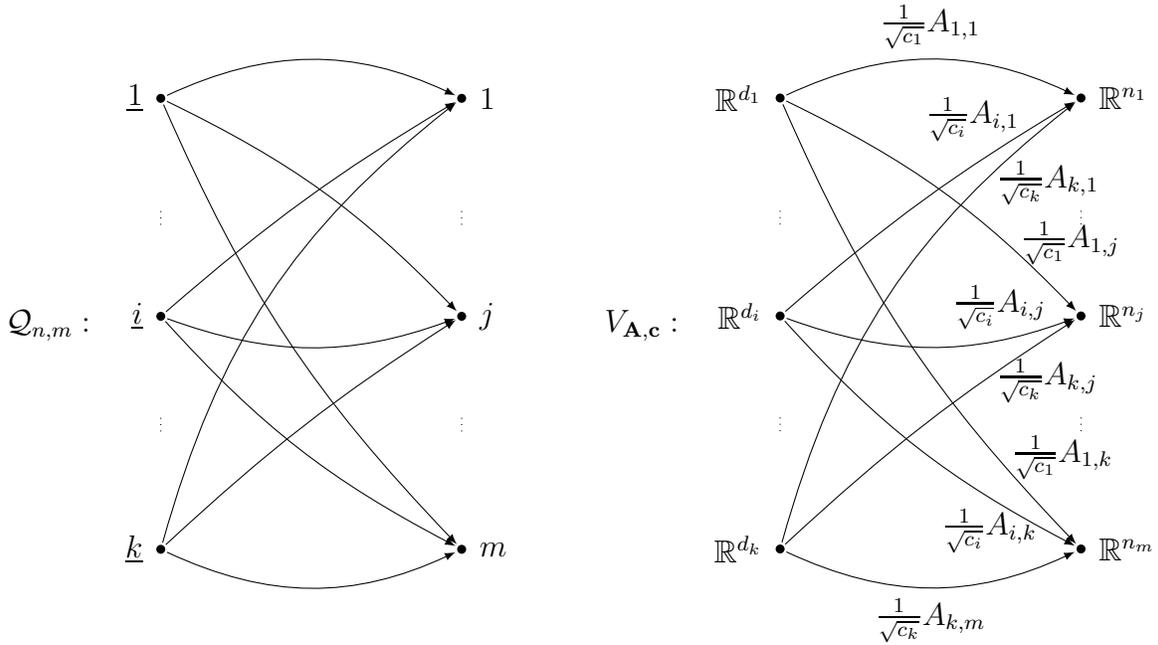
{\hbox{  
\begin{tikzpicture}[point/.style={shape=circle, fill=black, scale=.3pt,outer sep=3pt},>=latex]
   \node[point,label={left:$\underline{1}$}] (1) at (-4,3) {};
   \node[point,label={left:$\underline{i}$}] (2) at (-4,.1) {};
   \node[point,label={left:$\underline{k}$}] (3) at (-4,-3) {};
   
   \node[point,label={right:$1$}] (-1) at (0,3) {};
   \node[point,label={right:$j$}] (-2) at (0,.1) {};
   \node[point,label={right:$m$}] (-3) at (0,-3) {};
  
   \draw[dotted] (0,1.5)--(0,1.25);
   \draw[dotted] (-4,1.5)--(-4,1.25);
   \draw[dotted] (0,-1.25)--(0,-1.5);
   \draw[dotted] (-4,-1.25)--(-4,-1.5);
  
   \path[->]
  (1) edge [bend left=25]  (-1)
  (1) edge  [bend left=10] (-2)
  (1) edge  [bend right=10] (-3)
  (2) edge [bend left=5]  (-1)
  (2) edge [bend right=20](-2)
  (2) edge [bend right=10]  (-3)
  (3) edge [bend left=18]  (-1)
  (3) edge [bend left=5](-2)
  (3) edge  [bend right=25] (-3);
\end{tikzpicture} 
}}
\hspace{30pt}
V_{\A, \cc}:~
\vcenter{\hbox{
\begin{tikzpicture}[point/.style={shape=circle, fill=black, scale=.3pt,outer sep=3pt},>=latex]
   \node[point,label={left:$\RR^{d_1}$}] (1) at (-4,3) {};
   \node[point,label={left:$\RR^{d_i}$}] (2) at (-4,.1) {};
   \node[point,label={left:$\RR^{d_k}$}] (3) at (-4,-3) {};
   \node[point,label={right:$\RR^{n_1}$}] (-1) at (0,3) {};
   \node[point,label={right:$\RR^{n_j}$}] (-2) at (0,.1) {};
   \node[point,label={right:$\RR^{n_m}$}] (-3) at (0,-3) {};
   
   \draw[dotted] (0,1.5)--(0,1.25);
   \draw[dotted] (-4,1.5)--(-4,1.25);
   \draw[dotted] (0,-1.25)--(0,-1.5);
   \draw[dotted] (-4,-1.25)--(-4,-1.5);
  
   \path[->]
   (1) edge [bend left=25] node[above] {${1 \over \sqrt{c_1}}A_{1,1}$} (-1)
   (1) edge [bend left=10] node[very near end, above] {$\phantom{abcd} {1 \over \sqrt{c_1}}A_{1,j}$} (-2)
   (1) edge [bend right=10] node[very near end, above] {$\phantom{abcde} {1 \over \sqrt{c_1}}A_{1,k}$} (-3)
   (2) edge [bend left=5] node[near end, above] {${1 \over \sqrt{c_i}}A_{i,1}\phantom{abc}$} (-1)
   (2) edge [bend right=20] node[near end, above] {${1 \over \sqrt{c_i}}A_{i,j}$} (-2)
   (2) edge [bend right=10] node[near end, below] {${1 \over \sqrt{c_i}}A_{i, k}$} (-3)
   (3) edge [bend left=18] node[very near end, below] {$\phantom{abcd} {1 \over \sqrt{c_k}}A_{k,1}$} (-1)
   (3) edge [bend left=5] node[very near end, below] {$\phantom{ab} {1 \over \sqrt{c_k}}A_{k,j}$} (-2)
   (3) edge [bend right=25] node[below] {${1 \over \sqrt{c_k}}A_{k,m}$} (-3);  
\end{tikzpicture} 
}}
$$
Furthermore, the tuples $\cc$ and $\tup$ define the integral stability weight $\sigma_{\cc, \tup}$ of $\Q_{k,m}$ by assigning the integers $c_i$ and $-p_j$ to the vertices of $\Q_{k,m}$. 

In what follows, we briefly recall just enough terminology to state our main results, with more detailed background found in Section \ref{BL-operators-sec}. Let $\Q_0$ denote the set of vertices and $\Q_1$ denote the set of arrows of $\Q:=\Q_{k,m}$. A real \emph{representation} $\V$ of $\Q$ assigns a finite-dimensional real vector space $\V(x)$ to every vertex $x \in \Q_0$ and a linear map $\V(a): \V(ta) \to \V(ha)$ to every arrow $a \in \Q_1$. After fixing bases for the vector spaces $\V(x)$, $x \in \Q_0$, we often think of the linear maps $\V(a)$, $a \in \Q_1$, as matrices of appropriate size. The dimension vector of a representation $\V$ of $\Q$ is $\ddim \V:=(\dim_{\RR} V(x))_{x \in \Q_0} \in \NN^{\Q_0}$. 

A \emph{morphism} $\varphi:V \rightarrow W$ between two representations is a collection $(\varphi(x))_{x \in Q_0}$ of $\RR$-linear maps with $\varphi(x) \in \Hom_\RR(V(x), W(x))$ for each $x \in Q_0$, and such that $\varphi(ha) \circ V(a)=V(a) \circ \varphi(ta)$ for each $a \in Q_1$. For a representation $V$, the one-dimensional vector space $\{(\lambda \Id_{V(x)})_{x \in Q_0} \mid \lambda \in \RR\}$ is always a subspace of the space of endomorphisms, $\End_Q(V)$, of $V$. We say that $V$ is a \emph{Schur} representation if $\End_Q(V)$ is one-dimensional. 

The \emph{dimension vector} of an \emph{AJN-datum} $(\A,\cc, \tup)$ is defined to be $(\dd, \nn)$. We say that $(\A, \cc, \tup)$ is a \emph{Schur datum} if the corresponding representation $V_{\A, \cc}$ is a Schur representation.

Let $\sigma \in \ZZ^{\Q_0}$ be an integral weight of $\Q$ such that $\sigma$ is positive at the source verices and negative at the sink verices.  A representation $V$ of $\Q$ is said to be \emph{$\sigma$-semi-stable} if $\sigma \cdot \ddim V:=\sum_{x \in Q_0} \sigma(x)\dim V(x)=0$ and 
\begin{equation}\label{semi-stab-eqn}
\sum_{i=1}^k \sigma(\underline{i}) \dim_{\RR}V'(\underline{i})\leq \sum_{j=1}^m (-\sigma(j)) \dim_{\RR} \left( \sum_{i=1}^k V_{ij}(V'(\underline{i}))\right),
\end{equation}
for every collection of subspaces $\left(V'(\underline{i}) \subseteq V(\underline{i})\right)_{i \in [k]}$. (Here, $V_{ij}$ denotes the linear map $V(a)$ where $a$ is the unique arrow from the source vertex $\underline{i}$ to the sink vertex $j$.) We say that $V$ is \emph{$\sigma$-stable} if $\sigma \cdot \ddim V=0$ and $(\ref{semi-stab-eqn})$ holds strictly for all \emph{proper} collections of subspaces $\left(V'(\underline{i}) \subseteq V(\underline{i})\right)_{i \in [k]}$. We say that a representation is \emph{$\sigma$-polystable} if it is a finite direct sum of $\sigma$-stable representations. We say that an AJN-datum $(\A, \cc, \tup)$ is \emph{semi-stable/simple/semi-simple} if the corresponding representation $V_{\A, \cc}$ is $\sigma_{\cc, \tup}$-semi-stable/stable/polystable.

In \cite{ChiDer-2021}, guided by invariant theoretic considerations and \cite[Construction 4.2]{GarGurOliWig-2017}, we introduced the \emph{Brascamp-Lieb operator} $T_{V,\sigma}$ associated to the quiver datum $(V, \sigma)$ and defined the capacity of $(V, \sigma)$ to be the capacity of $T_{V, \sigma}$. For an AJN-datum $(\A, \cc, \tup)$, we write $\capa(\A, \cc, \tup)$ for the capacity of the Brascamp-Lieb operator $T_{V_{\A, \cc}, \sigma_{\cc, \tup}}$.

In general, the capacity of any (square) completely positive operator $T$ equals that of its dual $T^*$ (see \cite[Proposition 3.14]{GarGurOliWig-2017}). In Lemma \ref{lemma-compute-cap}, we show that when working with the dual operator $T^*_{V_\A, \sigma_{\cc, \tup}}$ instead of $T_{V_\A, \sigma_{\cc, \tup}}$, the AJN-constant $\mm(\A, \cc, \tup)$ can be expressed in terms of $\capa(\A, \cc, \tup)$ as follows:
\begin{equation} \label{capa-AJN-eqn}
e^{-2\mm(\A, \cc, \tup)}=\capa(\A,\cc, \tup)=\inf \left \{ {\prod_{j=1}^m \det \left( \sum_{i=1}^k A_{ij} \cdot \Sigma_i \cdot A_{ij}^T \right)^{p_j} \over \prod_{i=1}^k \det(\Sigma_i)^{c_i}} \;\middle|\; \Sigma_i \in \s^{+}_{d_i} \right \},
\end{equation}
and
\begin{equation}
\mm(\A, \cc, \tup) < \infty \Longleftrightarrow \capa(V_\A, \sigma_{\cc, \tup})>0.
\end{equation}
(Here, for any positive integer $d$, $\s^{+}_d$ denotes the set of all $d \times d$ (symmetric) positive definite real matrices.) 

Following \cite{ChiDer-2021}, we say that the quiver datum $(V_\A,\sigma_{\cc, \tup})$ is \emph{gaussian-extremizable} if the infimum in $(\ref{capa-AJN-eqn})$ is attained for some positive definite matrices $\Sigma_i \in \RR^{d_i}$, $i \in [k]$. If this is the case, we call such a $k$-tuple $(\Sigma_1, \ldots, \Sigma_k)$ a \emph{gaussian extremizer} for $(V_{\A, \cc},\sigma_{\cc, \tup})$. Another important concept introduced in \cite{ChiDer-2021} is that of a geometric quiver datum. Specifically, we say that $(V_{\A, \cc},\sigma_{\cc, \tup})$ is a \emph{geometric quiver datum} if the corresponding BL operator $T_{V_A, \sigma_{\cc, \tup}}$ is doubly-stochastic which simply means that
\begin{equation}
\sum_{j=1}^m p_j A_{ij}^T \cdot A_{ij}=c_i \Id_{d_i}, \forall i \in [k], \text{~and~}\sum_{i=1}^kA_{ij}\cdot A_{ij}^T=\Id_{n_j}, \forall j \in [m]. 
\end{equation}
Thus we have that
\begin{itemize}
\item $(V_{\A, \cc},\sigma_{\cc, \tup})$ is gaussian-extremizable if and only if $(\A, \cc, \tup)$ is gaussian-extremizable in the sense of \cite[Section III]{AJN-2022};

\item $(V_{\A, \cc},\sigma_{\cc, \tup})$ is a geometric quiver datum if and only if $(\A, \cc, \tup)$ is an AJN-geometric datum in the sense of \cite[Definition 5]{AraCouZha-2022}.
\end{itemize}
\noindent
One of the advantages of working with doubly stochastic completely positive operators is that their capacities are always one (see \cite[Proposition 2.8 and Lemma 3.4]{GarGurOliWig-2020}). Consequently, $\capa(\A, \cc, \tup)=1$ and $\mm(\A, \cc, \tup)=0$ for any AJN-geometric datum $(\A, \cc, \tup)$.

To state our main result, we need to recall a few more concepts. The base change group $\mathbf{G}:=\prod_{i=1}^k \GL(d_i)\times \prod_{j=1}^m \GL(n_j)$ acts on the space of tuples $B=(B_{ij})_{i \in [k], j \in [m]}$ with $B_{ij} \in \RR^{n_j \times d_i}$ by simultaneous conjugation, \emph{i.e.}, for any $g=(g(x))_{x \in Q_0} \in \mathbf{G}$,  we have that
$$
g \cdot B=(g(j)\cdot B_{ij}\cdot g(\underline{i})^{-1})_{i \in [k], j\in [m]}.
$$

The \emph{character} of $\mathbf{G}$ induced by the exponents $(\cc, \tup)$ is the rational character $\chi_{\cc, \tup}:\mathbf{G} \to \RR^{\times}=\RR \setminus \{0\}$ defined by 
$$
\chi_{\cc, \tup}(g)=\prod_{i\in [k]}\det(g(\underline{i}))^{c_i} \cdot \prod_{j \in [m]}\det(g(j))^{-p_j}, 
$$
for all $g=(g(x))_{x \in Q_0} \in \mathbf{G}$. We say that an AJN-datum $(\widetilde{\A}, \cc, \tup)$ of dimension vector $(\dd, \nn)$ is a \emph{degeneration} of $(\A, \cc, \tup)$ if the representation $V_{\widetilde{\A}, \cc}$ lies in the closure of the $\mathbf{G}_{\cc, \tup}$-orbit of $V_{\A, \cc}$ where $\mathbf{G}_{\cc, \tup}$ is the kernel of the character $\chi_{\cc, \tup}$. 

We are now ready to state our main result.

\begin{theorem} \label{main-thm} Let $(\A, \cc, \tup)$ be an AJN-datum and let $\mm(\A, \cc, \tup)$ be the corresponding AJN-constant. 

\begin{enumerate}
\item $(\A, \cc, \tup)$ is feasible if and only if it is semi-stable.
\bigskip

\item (\textbf{Kempf-Ness Theorem for AJN data}) The datum $(\A\cc, \tup)$ is semi-simple if and only if $\G(\A, \cc, \tup) \neq \emptyset$ where
$$
\G(\A, \cc, \tup):=\{g \in \mathbf{G} \mid (g\cdot \A, \cc, \tup) \text{~is an AJN-geometric datum}\}.
$$
\smallskip

\item (\textbf{Character formula for AJN constants}) Assume that $(\A, \cc, \tup)$ is a feasible datum. Then there exists a semi-simple degeneration $(\widetilde{\A}, \cc, \tup)$ of $(\A, \cc, \tup)$. Furthermore, for any such degeneration $(\widetilde{\A}, \cc, \tup)$, the following formula holds
$$
\mm(\A, \cc, \tup)=-\log(|\chi_{\cc, \tup}(g)|), \forall g \in \G(\widetilde{\A}, \cc, \tup).
$$ 	
\smallskip

\item (\textbf{Decomposition of AJN constants}) Assume that 
$$
A_{ij}=\left(
\begin{matrix}
B_{ij} & X_{ij}\\
0&C_{ij}
\end{matrix}\right) \in \RR^{n_j\times d_i},
$$
where $B_{ij} \in \RR^{n^1_j\times d^1_i}$, $C_{ij} \in \RR^{n^2_j\times d^2_i}$, and $X_{ij} \in \RR^{n^1_j\times d^2_i}$. If 
$$
\sum_{i=1}^k c_i d^1_i=\sum_{j=1}^m p_j n^1_j
$$
then
$$
\mm(\A, \cc, \tup)=\mm(\B,\cc, \tup)+ \mm(\C,\cc, \tup).
$$

\bigskip		
		
\item (\textbf{Gaussian extremizers: existence}) The datum $(\A, \cc, \tup)$ is gaussian-extremizable if and only if $(\A, \cc, \tup)$ is semi-simple. If this is the case then the gaussian extremizers of $(\A, \cc, \tup)$ are the $k$-tuples of matrices 
$$
(g(\underline{i})^{-1} \cdot g(\underline{i})^{-T})_{i \in [k]} \text{~with~} g \in \G(\A, \cc, \tup).$$ 

\bigskip

\item (\textbf{Gaussian extremizers: uniqueness}) (a) If $(\A, \cc, \tup)$ has unique gaussian extremizers (up to scaling) then $(\A, \cc, \tup)$ is a simple datum. \\

\noindent
(b) If $(\A, \cc, \tup)$ is a Schur simple datum then $(\A, \cc, \tup)$ has unique gaussian extremizers (up to scaling).

\end{enumerate}
\end{theorem}

Theorem \ref{main-thm}{(1)} has also been proved in \cite{AJN-2022} and \cite{Cou-2019}. Parts $(2)$ and $(5)$ of Theorem \ref{main-thm} show that $(\A, \cc, \tup)$ is gaussian-extremizable if and only if $(g \cdot \A, \cc, \tup)$ is an AJN-geometric datum for some $g \in \mathbf{G}$. This equivalence has also been proved in \cite{AraCouZha-2022}. Nonetheless, our quiver invariant theoretic methods are very different than those in \emph{loc. cit.}

\section{Background on Quiver Invariant Theory} \label{BL-operators-sec}
Throughout, we work over the field $\RR$ of real numbers and denote by $\NN=\{0,1,\dots \}$. For a positive integer $L$, we denote by $[L]=\{1, \ldots, L\}$.

A quiver $Q=(Q_0,Q_1,t,h)$ consists of two finite sets $Q_0$ (\emph{vertices}) and $Q_1$ (\emph{arrows}) together with two maps $t:Q_1 \to Q_0$ (\emph{tail}) and $h:Q_1 \to Q_0$ (\emph{head}). We represent $Q$ as a directed graph with set of vertices $Q_0$ and directed edges $a:ta \to ha$ for every $a \in Q_1$.  We assume throughout that $Q$ is a connected quiver, \emph{i.e.}, its underlying graph is connected.

A representation of $Q$ is a family $\V=(\V(x), \V(a))_{x \in Q_0, a\in Q_1}$ where $\\V(x)$ is a finite-dimensional $\RR$-vector space for every $x \in Q_0$, and $\V(a): \V(ta) \to \V(ha)$ is a $\RR$-linear map for every $a \in Q_1$. A \emph{subrepresentation} $\V'$ of $\V$, written as $V' \leq \V$, is a representation of $Q$ such that $\V'(x) \subseteq \V(x)$ for every $x \in Q_0$, and $\V(a)(\V'(ta)) \subseteq \V'(ha)$ and $\V'(a)$ is the restriction of $\V(a)$ to $\V(ta)$ for every arrow $a \in Q_1$. 

A morphism $\varphi:\V \rightarrow W$ between two representations is a collection $(\varphi(x))_{x \in Q_0}$ of $\RR$-linear maps with $\varphi(x) \in \Hom_K(\V(x), W(x))$ for every $x \in Q_0$, and such that $\varphi(ha) \circ V(a)=\V(a) \circ \varphi(ta)$ for every $a \in Q_1$. For a representation $\V$, the one-dimensional vector space $\{(\lambda \Id_{\V(x)})_{x \in Q_0} \mid \lambda \in \RR\}$ is always a subspace of the space of endomorphisms, $\End_Q(V)$, of $\V$. We say that $\V$ is a \emph{Schur representation} if $\End_Q(\V)$ is one-dimensional.

The dimension vector $\ddim \V \in \NN^{Q_0}$ of a representation $\V$  is defined by $\ddim \V(x)=\dim_{\RR} \V(x)$ for all $x \in Q_0$. By a dimension vector of $Q$, we simply mean a $\ZZ_{\geq 0}$-valued function on the set of vertices $Q_0$. For two vectors $\theta, \beta \in \RR^{Q_0}$, we define $\theta \cdot \beta=\sum_{x \in Q_0} \theta(x)\beta(x)$. 

Let $\beta \in \NN^{Q_0}$ be a dimension vector. The representation space of $\beta$-dimensional representations of $Q$ is the affine space
$$
\rep(Q,\beta)=\prod_{a \in Q_1}\RR^{\beta(ha)\times \beta(ta)}.
$$
The base change group $\GL(\beta)=\prod_{x \in Q_0}\GL(\beta(x), \RR)$ acts on $\rep(Q,\beta)$ by simultaneous conjugation, i.e. for $g=(g(x))_{x \in Q_0}$ and $\V=(\V(a))_{a \in Q_1}$, we have that
$$(g\cdot \V)(a)=g(ha)\cdot \V(a) \cdot g(ta)^{-1}, \forall a \in Q_1.
$$
Note that there is a bijective correspondence between the isomorphism classes of representations of $Q$ of dimension vector $\beta$ and the $\GL(\beta)$-orbits in $\rep(Q,\beta)$.

\subsection{The capacity of Brscamp-Lieb operators associated to quiver data} From now on, we assume that $Q$ is a (not necessarily complete) bipartite quiver, \emph{i.e.}, $Q_0$ is the disjoint union of two subsets $Q^+_0$ and $Q^-_0$, and all the arrows in $Q$ go from $Q^+_0$ to $Q^-_0$. Write $Q_0^{+}=\{v_1, \ldots, v_k\}$ and $Q_0^{-}=\{w_1,\ldots, w_m\}$.  

Let us fix an integral weight $\sigma \in \ZZ^{Q_0}$ such that $\sigma$ is positive on $Q^+_0$ and negative on $Q^-_0$. Define
$$\sigma_{+}(v_i)=\sigma(v_i), \forall i \in [k], \text{~and~}
\sigma_{-}(w_j)=-\sigma(w_j), \forall j \in [m].
$$
Let us assume that $\sigma \cdot \beta=0$ which is equivalent to
$$
N:=\sum_{i=1}^k \sigma_{+}(v_i)\beta(v_i)=\sum_{j=1}^m \sigma_{-}(w_j)\beta(w_j).
$$

For $i\in [k]$ and $j \in [m]$, we denote the set of all arrows in $Q$ from $v_i$ to $w_j$  by $\ar_{ij}$. If there are no arrows from $v_i$ to $w_j$, we define $\ar_{ij}$ to be the set consisting of the symbol $\0_{ij}$.

Let $M:=\sum_{j=1}^m \sigma_{-}(w_j) \text{~and~} M':=\sum_{i=1}^k \sigma_{+}(v_i).$
For each $j \in [m]$ and $i \in [k]$, define 

$$
\jj:=\{q \in \ZZ \mid \sum_{l=1}^{j-1} \sigma_{-}(w_l) < q \leq  \sum_{l=1}^{j} \sigma_{-}(w_l) \},
$$
and 
$$
\ii:=\{r \in \ZZ \mid \sum_{l=1}^{i-1} \sigma_{+}(v_l) < r \leq  \sum_{l=1}^{i} \sigma_{+}(v_l) \}.
$$
In what follows, we consider $M \times M'$ block matrices of size $N \times N$ such that for any two indices $q \in \jj$ and $r \in \ii$, the $(q,r)$-block-entry is a matrix of size $\beta(w_j)\times \beta(v_i)$. Set
$$
\s:= \{(i,j,a,q,r) \mid i \in [k], j \in [m],\\ a \in \ar_{ij}, \\ q \in \jj, r \in \ii \}.
$$

\smallskip
Now, let $\V\in \rep(Q,\beta)$ be a $\beta$-dimensional representation of $Q$. For each $(i,j,a,q,r) \in \s$, let $\V^{ij,a}_{q,r}$ be the $M \times M'$ block matrix whose $(q,r)$-block-entry is $\V(a) \in \RR^{\beta(w_j)\times \beta(v_i)}$, and all other entries are zero. The convention is that if $a=\0_{ij} \in \ar_{ij}$ then $\V(a)$ is the zero matrix of size $\beta(w_j)\times \beta(v_i)$; hence, if there are no arrows from $v_i$ to $w_j$ then $\V^{i,j,a}_{q,r}$ is the zero matrix of size $N \times N$.

\begin{definition} \label{BL-operator-def} Let $\V \in \rep(Q,\beta)$ be a $\beta$-dimensional representation of $Q$.
\begin{enumerate}[(i)]
\item The \emph{Brascamp-Lieb} operator $T_{\V, \sigma}$ associated to $(\V, \sigma)$ is defined to be the completely positive operator with Kraus operators $\V^{i,j,a}_{q,r}$, $(i,j,a,q,r) \in \s$, \emph{i.e.},
\begin{align*}
T_{(\V,\sigma)}: \RR^{N \times N} & \to \RR^{N \times N}\\
X& \to T_{\V,\sigma}(X):=\sum_{(i,j,a,q,r)}(\V^{i,j,a}_{q,r})^T\cdot X \cdot \V^{i,j,a}_{q,r}
\end{align*}

\item The \emph{capacity} $\capa(\V,\sigma)$ of $(\V,\sigma)$ is defined to be the capacity of $T_{V, \sigma}$, \emph{i.e.},
$$
\capa(\V,\sigma):=\inf \{\Det(T_{\V,\sigma}(X)) \mid  X \in \s^{+}_N,  \Det(X)=1 \}.
$$
(Here, for a given positive integer $d$, we denote by $\s^{+}_d$ the set of all $d \times d$ (symmetric) positive definite real matrices.) 

\bigskip
\noindent
For an AJN-datum $(\A, \cc, \tup)$, we define
\begin{equation}
\capa(\A, \cc, \tup):=\capa(\V_\A, \sigma_{\cc, \tup}).
\end{equation}
\end{enumerate}
\end{definition}

\bigskip
\noindent
In \cite{ChiDer-2021}, we showed that
\begin{equation}
\capa(\V,\sigma)=\inf\left \{ { \prod_{i=1}^n \det \left( \sum_{j=1}^m \sigma_{-}(w_j)\left( \sum_{a \in \ar_{ij}} \V(a)^T \cdot Y_j \cdot \V(a)  \right) \right)^{\sigma_{+}(v_i)}   \over \prod_{j=1}^m \det(Y_j)^{\sigma_{-}(w_j)}} \right \},
\end{equation}
where the infimum is over all positive definite matrices $Y_j \in \s^{+}_{\beta(w_j)}, j \in [m]$.    
On the other hand, the capacity of a (square) completely positive operator can be computed as the capacity of its dual. Thus 
\begin{equation}
\capa(\V, \sigma)=\capa(T^*_{\V, \sigma}),
\end{equation}
where $T_{\V,\sigma}^*(X):=\sum_{(i,j,a,q,r)}\V^{i,j,a}_{q,r}\cdot X \cdot (\V^{i,j,a}_{q,r})^T$ is the dual operator of $T_{\V, \sigma}$. Working with $T^*_{\V, \sigma}$ instead of $T_{\V, \sigma}$, we obtain the following formula for $\capa(\V, \sigma)$.

\begin{lemma} (compare to \cite[Lemma 8]{ChiDer-2021}) \label{lemma-compute-cap} Let $(\V, \sigma)$ be a quiver datum with $\V \in \rep(Q,\beta)$. Then
\begin{align*}
&\capa(\V,\sigma)= \\
&=\inf\left \{ { \prod_{j=1}^m \det \left( \sum_{i=1}^k \sigma_{+}(v_i)\left( \sum_{a \in \ar_{ij}} \V(a) \cdot \Sigma_i \cdot \V(a)^T  \right) \right)^{\sigma_{-}(w_j)}   \over \prod_{i=1}^k \det(\Sigma_i)^{\sigma_{+}(v_i)}} \;\middle|\; \Sigma_i \in \s^{+}_{\beta(v_i)}    \right \}.
\end{align*}

Consequently, for an AJN-datum $(\A, \cc, \tup)$,
$$
\capa(\A,\cc, \tup)=\inf \left \{ {\prod_{j=1}^m \det \left( \sum_{i=1}^k A_{ij} \cdot \Sigma_i \cdot A_{ij}^T \right)^{p_j} \over \prod_{i=1}^k \det(\Sigma_i)^{c_i}} \;\middle|\; \Sigma_i \in \s^{+}_{d_i} \right \},
$$
and hence 
$$\capa(\A, \cc, \tup)=e^{-2 \mm(\A, \cc, \tup)}.
$$
\end{lemma}

\begin{proof} We have that
$$
T^*_{\V, \sigma}(X)=\sum_{(i,j,a,q,r)}\V^{i,j,a}_{q,r}\cdot X \cdot (\V^{i,j,a}_{q,r})^T, \forall X \in \RR^{N \times N}.
$$ 
In what follows, we view an $N \times N$ matrix $X$ as an $M' \times M'$ block matrix. Furthermore, for every $i \in [k]$ and $r \in \ii$, denote the $(r,r)$-block-diagonal entry of $X$ by $X_{rr}$.  Thus, for any $(i,j,a,q,r) \in \s$, the matrix 
$$
\V^{i,j,a}_{q,r} \cdot X \cdot (\V^{i,j,a}_{q,r})^T
$$ 
has an $M \times M$ block matrix structure whose $(q,q)$-block entry is
$$
\V(a)\cdot X_{rr}\cdot (\V(a))^T,
$$
and all other blocks are zero. So, $T^*_{\V, \sigma}(X)$ is the $M \times M$ block-diagonal matrix whose $(q,q)$-block-diagonal entry is
$$
\sum_{i=1}^n \sum_{a \in \ar_{ij}} \V(a) (\sum_{r \in \ii} X_{rr}) \V(a)^T,
$$
for all $q \in \jj$ and $j \in [m]$. We we can write

\begin{align*}
&\capa(\V, \sigma)= \inf \{ \det(T^*_{\V,\sigma}(X)) \mid X \in \s^+_N, \det(X)=1 \}\\
=&\inf \left \{ \prod_{j=1}^m \det \left( \sum_{i=1}^k \sum_{a \in \ar_{ij}}  \V(a) \left( \sum_{r \in \ii} X_{rr} \right) \V(a)^T   \right)^{\sigma_{-}(w_j)} \;\middle|\; X \in \s^{+}_{N} , \det(X)=1 \right \} \\
\stackrel{(iii)}{=}&\inf \left \{ \prod_{j=1}^m \det \left( \sum_{i=1}^k \sum_{a \in \ar_{ij}}  \V(a) \left( \sum_{r \in \ii} X_r \right) \V(a)^T   \right)^{\sigma_{-}(w_j)} \;\middle|\; X_r \in \s^{+}_{\beta(v_i)},  \prod_{i=1}^k \prod_{r \in \ii}\det(X_r)=1 \right \} \\
\stackrel{(iv)}{=}&\inf \left \{ \prod_{j=1}^m \det \left( \sum_{i=1}^k\sigma_{+}(v_i) \left( \sum_{a \in \ar_{ij}}  \V(a) \cdot \Sigma_i \cdot \V(a)^T \right)  \right)^{\sigma_{-}(w_j)} \;\middle|\; \Sigma_i \in \s^{+}_{\beta(v_i)},  \prod_{i=1}^k \det(\Sigma_i)^{\sigma_{+}(v_i)}=1 \right \} \\
\stackrel{(v)}{=}&\inf \left \{ { \prod_{j=1}^m \det \left( \sum_{i=1}^k \sigma_{+}(v_i) \left( \sum_{a \in \ar_{ij}}  \V(a) \cdot \Sigma_i \cdot \V(a)^T \right)   \right)^{\sigma_{-}(w_j)} \over \prod_{i=1}^k \det(\Sigma_i)^{\sigma_{+}(v_i)} } \;\middle|\; \Sigma_i \in \s^{+}_{\beta(v_i)} \right \}
\end{align*}
To prove $(iii)$ above, it is clear that the infimum displayed on the second line above is less than or equal to that on the third line. To prove the reverse inequality, let $X$ be a positive definite $N \times N$ real matrix with $\det(X)=1$. Denoting its block-diagonal entries by $X_r$, $r \in \ii$, $i \in [k]$, we get that
$$
1=\det(X) \leq C:=\prod_{i=1}^k \prod_{r \in \ii}\det(X_r).
$$ 
Setting $\widetilde{X}_r={1 \over \sqrt[N]{C}}X_r$, we get that $\prod_{i=1}^m \prod_{r \in \ii}\det(\widetilde{X}_r)=1$, and
\begin{align*}
\prod_{j=1}^m \det &  \left( \sum_{i=1}^k \sum_{a \in \ar_{ij}}  \V(a) \left( \sum_{r\in \ii} \widetilde{X}_r \right)  \V(a)^T \right)^{\sigma_{-}(w_j)}=\\
&={1 \over C} \prod_{j=1}^m \det \left( \sum_{i=1}^k \sum_{a \in \ar_{ij}}  \V(a) \left( \sum_{r \in \ii} X_r \right) \V(a)^T   \right)^{\sigma_{-}(w_j)}. 
\end{align*}
This now gives get the reverse inequality, proving the third equality above. For $(iv)$, it is clear that the infimum on the line above is less than or equal to that on the fourth line. To prove the reverse inequality, let $X_r \in \s^{+}_{\beta(v_i)}$, $r \in \ii$, $i \in [k]$, be matrices such that $\prod_{i=1}^k \prod_{r \in \ii}\det(X_r)=1$.  Using the generalized Hadamard's inequality, we have
$$
\prod_{r \in \ii} \det(X_r) \leq \det(\widetilde{\Sigma}_i)^{\sigma_{+}(v_i)} \text{~where~}\widetilde{\Sigma}_i:={\sum_{r \in \ii} X_r \over \sigma_{+}(v_i)} \text{~for all~} i \in [k].
$$
Now let $D>0$ be such that $D^N=\prod_{i\in [k]} \det(\widetilde{\Sigma}_i)^{\sigma_{+}(v_i)}$. Then $D \geq 1$ and 
$$
\prod_{i \in [k]} \det(D^{-1} \widetilde{\Sigma}_i)^{\sigma_{+}(v_i)}=1.
$$
Finally, setting $\Sigma_i:=D^{-1}\widetilde{\Sigma}_i$, $i \in [k]$,  we get
\begin{align*}
\prod_{j=1}^m \det \left( \sum_{i=1}^k \sum_{a \in \ar_{ij}}  \V(a) \left( \sum_{r \in \ii} X_r \right) \V(a)^T   \right)^{\sigma_{-}(w_j)} \geq \\ 
\geq \prod_{j=1}^m  \det \left( \sum_{i=1}^k\sigma_{+}(v_i) \left( \sum_{a \in \ar_{ij}}  \V(a) \cdot \Sigma_i \cdot \V(a)^T \right)  \right)^{\sigma_{-}(w_j)} &,
\end{align*}
which proves $(iv)$. For $(v)$, simply work with ${\Sigma_i \over \sqrt[N]{\prod_{i=1}^k \det(\Sigma_i)^{\sigma_{+}(v_i)}}}$, $i \in [k]$, in the line above where $\Sigma_i \in \Mat_{\dd(v_i)\times \dd(v_i)}$, $i \in [k]$, are arbitrary positive definite matrices.
\end{proof}

Let $\chi_{\sigma}:\GL(\beta)\to \RR^{\times}$ be the character induced by $\sigma$, i.e. $\chi_{\sigma}(g)=\prod_{x \in Q_0}\det(g(x))^{\sigma(x)}$ for all $g=(g(x))_{x \in Q_0} \in \GL(\beta)$, and denote its kernel by $\GL(\beta)_{\sigma}$. As a consequence of the lemma above, we get the following formula for the capacity along $\GL(\beta)$-orbits.

\begin{corollary} \label{capa-char-formula-coro} Let $(\V, \sigma)$ be a quiver datum with $\V \in \rep(Q,\beta)$. Then 
$$
\capa(\V, \sigma)=(\chi_{\sigma}(g))^2 \cdot \capa(g \cdot \V, \sigma), \forall g=(g(x))_{x \in Q_0} \in \GL(\beta).
$$
In particular, if $g \in \GL(\beta)_{\sigma}$ then $$\capa(\V,\sigma)=\capa(g \cdot \V, \sigma),$$ i.e. the capacity is constant along $\GL(\beta)_{\sigma}$-orbits.
\end{corollary}

\subsection{Semi-stability of a quiver datum and the positivity of its capacity}
Let $\V$ be a representation and $\sigma \in \ZZ^{Q_0}$ an integral weight of an (arbitrary) quiver $Q$. We say that $\V$ is \emph{$\sigma$-semi-stable} if
\begin{equation} \label{semi-stab-defn}
\sigma \cdot \ddim \V=0 \text{~and~} \sigma \cdot \ddim V' \leq 0,
\end{equation}
for all subrepresentations $\V'$ of $\V$. We say that $\V$ is \emph{$\sigma$-stable} if $(\ref{semi-stab-defn})$ holds with strict inequalities for all proper subrepresentations $\V'$ of $ \V$ with $\V' \neq 0, \V$.  A representation is said to be \emph{$\sigma$-polystable} if it is a direct sum of $\sigma$-stable representations.

\begin{rmk} It is immediate to see that any \emph{complex} $\sigma$-stable representation is a Schur representations. However this does not hold for real representations. This is the reason we need to assume that the simple datum is a Schur datum in part $6(b)$ of Theorem \ref{main-thm}.
\end{rmk}

\begin{rmk} For a bipartite quiver $Q$, it is easy to see that a representation $\V$ is $\sigma$-semi-stable if and only if $\sigma \cdot \ddim \V=0$ and
\begin{equation} \label{semi-stab-ineq}
\sum_{i=1}^k \sigma_{+}(v_i) \dim \V'(v_i) \leq \sum_{j=1}^m \sigma_{-}(w_j) \dim \left( \sum_{i=1}^k \sum_{a \in \ar_{ij}}\V(a)(\V'(v_i)) \right),
\end{equation}
for all subspaces $\V'(v_i) \leq \RR^{\beta(v_i)}$, $\forall i \in [k]$. 
\end{rmk}

\begin{theorem}\cite[Theorem 1]{ChiDer-2021} \label{semi-stab-cap-thm} Let $(\V, \sigma)$ be a quiver datum with $\V \in \rep(Q, \beta)$. Then
$$
\V \text{~is~} \sigma-\text{semi-stable} \Longleftrightarrow \capa(\V, \sigma)>0.
$$
Consequently, for an AJN-datum $(\A, \cc, \tup)$, the AJN-constant $\mm(\A, \cc, \tup)$ is finite if and only if $(\A, \cc, \tup)$ is semi-stable.
\end{theorem}

\subsection{Geometric quiver data} \label{geo-quiver-data-sec}
Let $(\V, \sigma)$ be a quiver datum with $\V \in \rep(Q,\beta)$ and let $T_{\V,\sigma}$ be the Brascamp-Lieb operator associated to $(\V,\sigma)$. By definition, $T_{\V,\sigma}$ is a \emph{doubly stochastic} operator if $T_{\V,\sigma}(\Id)=T^*_{\V,\sigma}(\Id)=\Id$. This is easily seen to be equivalent to
\begin{equation}\label{geom-eq-1}
\sum_{j=1}^m \sigma_{-}(w_j) \sum_{a \in \ar_{ij}} \V(a)^T \cdot \V(a)=\Id_{\beta(v_i)}, \forall i \in [k],
\end{equation}
and
\begin{equation}\label{geom-eq-2}
\sum_{i=1}^k \sigma_+(v_i) \sum_{a \in \ar_{ij}} \V(a)\cdot \V(a)^T=\Id_{\beta(w_j)}, \forall j \in [m].
\end{equation}

\begin{definition} (see \cite[Definition 12]{ChiDer-2021}) We say that $(V,\sigma)$ is a \emph{geometric quiver datum} if $V$ satisfies the matrix equations $(\ref{geom-eq-1})$ and $(\ref{geom-eq-2})$. 
\end{definition}
Consequently, for an AJN-datum $(\A, \cc, \tup)$, the corresponding quiver datum $(\V_\A, \sigma_{\cc, \tup})$ is geometric if and only if 
$$
\sum_{j=1}^m p_j A_{ij}^T \cdot A_{ij}=c_i \Id_{d_i}, \forall i \in [k], \text{~and~}\sum_{i=1}^kA_{ij}\cdot A_{ij}^T=\Id_{n_j}, \forall j \in [m]. 
$$

\begin{rmk}It follows from \cite[Proposition 2.8 and Lemma 3.4]{GarGurOliWig-2015}) that the capacity of any doubly stochastic operator is always one. Thus $\capa(\V,\sigma)=1$ for any geometric datum $(\V,\sigma)$.
\end{rmk}

The next result, proved in \cite{ChiDer-2021}, captures the matrix equations $(\ref{geom-eq-1})$ and $(\ref{geom-eq-2})$ in the context of quiver invariant theory.  

\begin{theorem} (see \cite[Theorem 2]{ChiDer-2021}) \label{quiver-geom-data-thm}  Let $\V \in \rep(Q,\beta)$ be a $\beta$-dimensional representation of $Q$ and let 
$$
\G_{\sigma}(\V):=\{g \in \GL(\beta) \mid (g \cdot \V, \sigma)\text{~is a geometric quiver datum}\}.
$$ 
Then the following statements hold.

\begin{enumerate}[(a)]
\item  $\G_{\sigma}(\V) \neq \emptyset$ if and only if $\V $ is $\sigma$-polystable.

\item Assume that $V$ is $\sigma$-semi-stable. Then there exists a $\sigma$-polystable representation $\widetilde{\V}$ such that $\widetilde{\V} \in \overline{\GL(\beta)_{\sigma}\V}$. Furthermore, for any such $\widetilde{\V}$, the following formula holds:
\begin{equation}\label{capa-formula-eqn}
\capa(\V,\sigma)=\capa(\widetilde{\V},\sigma)=\chi_{\sigma}(g)^2, \forall g \in \G_{\sigma}(\widetilde{\V}).
\end{equation}

\item Assume that $\V$ is such that
$$
\V(a)=\left(
\begin{matrix}
\V_1(a) & X(a)\\
0&\V_2(a)
\end{matrix}\right),
\forall a \in Q_1,
$$
where $\V_i \in \rep(Q,\beta_i)$, $i \in \{1,2\}$, are representations of $Q$, and $X(a) \in \RR^{\beta_1(ha)\times \beta_2(ta)}, \forall a \in Q_1$. If $\sigma \cdot  \ddim \V_1=0$ then
$$
\capa(\V,\sigma)=\capa(\V_1,\sigma)\cdot \capa(\V_2,\sigma).
$$
\end{enumerate}
\end{theorem}
\bigskip

\section{Extremizable quiver data} 

Let $Q=(Q_0,Q_1,t,h)$ be a bipartite quiver with set of source vertices $Q_0^{+}=\{v_1, \ldots, v_k\}$ and set of sink vertices $Q_0^{-}=\{w_1,\ldots, w_m\}$. Let $\beta \in \ZZ_{>0}^{Q_0}$ be a dimension vector of $Q$ and let $\sigma \in \ZZ^{Q_0}$ be an integral weight such that $\sigma \cdot \beta=0$. Furthermore, assume that $\sigma$ is positive on $Q_0^+$ and negative on $Q_0^-$. Recall that $\sigma_{+}(v_i)=\sigma(v_i), \forall i \in [k]$, and $\sigma_{-}(w_j)=-\sigma(w_j), \forall j \in [m]$.

For a representation $\V \in \rep(Q,\beta)$ and a $k$-tuple $\Sigma=(\Sigma_1,\ldots, \Sigma_k)$ with $Y_i \in \s^{+}_{\beta(v_i)}$, $i \in [k]$, we set
$$
\capa(\V, \sigma; \Sigma):={ \prod_{j=1}^m \det \left( \sum_{i=1}^k \sigma_{+}(v_i)\left( \sum_{a \in \ar_{ij}} \V(a) \cdot \Sigma_i \cdot \V(a)^T  \right) \right)^{\sigma_{-}(w_j)}   \over \prod_{i=1}^k \det(\Sigma_i)^{\sigma_{+}(v_i)}}
$$

\begin{definition}  (compare to \cite[Definition 18]{ChiDer-2021}) Let $\V \in \rep(Q,\beta)$ be a $\beta$-dimensional representation. We say that $(\V, \sigma)$ is \emph{gaussian-extremizable} if there exists a $k$-tuple $\Sigma=(\Sigma_i)_{i=1}^k$ with $\Sigma_i \in \s^{+}_{\beta(v_i)}$, $i \in [k]$, such that 
$$
\capa(\V,\sigma)=\capa(\V,\sigma; \Sigma).
$$

\noindent
We call any such tuple $\Sigma$ a \emph{gaussian extremizer} for $(\V,\sigma)$.
\end{definition}

In what follows, we write $0\neq V \in \rep(Q, \beta)$ to mean that $V$ is not the zero vector of the vector space $\rep(Q, \beta)$, \emph{i.e.}, if $V(a) \in \RR^{\beta(ha)\times \beta(ta)}$ is not the zero matrix for some arrow $a \in Q_1$.

\begin{remark} \label{gauss-ext-pos-cap} If $(V, \sigma)$ is gaussian-extremizable and $0 \neq V \in \rep(Q, \beta)$ then it is immediate to see that $\capa(V, \sigma)>0$.
\end{remark}

\begin{remark} Let $(\V,\sigma)$ be a gaussian-extremizable quiver datum with gaussian extremizer $\Sigma=(\Sigma_i)_{i=1}^k$. We claim that for any $g=(g(x))_{x \in Q_0} \in \GL(\beta)$, $(g \cdot \V, \sigma)$ is gaussian-extremizable with gaussian extremizer
\begin{equation}\label{eqn-g-extremizers}
\widetilde{\Sigma}:=(g(v_i)\cdot \Sigma_i \cdot g(v_i)^{T} )_{i \in [k]}.
\end{equation}
Indeed, it is straightforward to see that 
$$\capa(\V,\sigma)=\capa(V,\sigma; \Sigma)=(\chi_{\sigma}(g))^2\cdot \capa(g \cdot \V,\sigma;\widetilde{\Sigma}).$$ 
Using Corollary \ref{capa-char-formula-coro}, we then get that 
$$\capa(g \cdot \V,\sigma)=(\chi_{\sigma}(g))^{-2}\cdot \capa(\V, \sigma)=\capa(g \cdot \V,\sigma;\widetilde{\Sigma}),$$
and this proves our claim.
\end{remark}

The next result gives necessary and sufficient conditions for $(V, \sigma)$ to be gaussian-extremizable and explains how to construct all the gaussian extremizers from $\G_{\sigma}(\V)$.

\begin{theorem} (compare to \cite[Theorem 20]{ChiDer-2021}) \label{gaussian-extremizers-thm} Let $(\V, \sigma)$ be a quiver datum with $0 \neq V \in \rep(Q,\beta)$. Then the following statements are equivalent:

\begin{enumerate}
\item $\V$ is $\sigma$-polystable;

\item $(\V, \sigma)$ is gaussian-extremizable.
\end{enumerate}

If either $(1)$ or $(2)$ holds then gaussian extremizers for $(\V, \sigma)$ are the $k$-tuples 
$$
(g(v_i)^{-1} \cdot g(v_i)^{-T})_{i\in [k]} \text{~with~}g \in \G_{\sigma}(\V). 
$$
\end{theorem}

\begin{proof} $(\Longrightarrow)$ Assume that $\V$ is $\sigma$-polystable. Then, by Theorem \ref{quiver-geom-data-thm}, there exists $g \in \GL(\beta)$ such that $(g \cdot \V,\sigma)$ is geometric datum; in particular, $(g \cdot V, \sigma)$ is gaussian-extremizable with gaussian extremizer $(\Id_{\beta(v_i)})_{i=1}^k$. This observation combined with $(\ref{eqn-g-extremizers})$ shows that $(\V,\sigma)$ is gaussian-extremizable with gaussian extremizer $(g(v_i)^{-1} \cdot g(v_i)^{-T})_{i \in [k]}$.  

\smallskip
\noindent
$(\Longleftarrow)$ Let us assume now that $(\V, \sigma)$ is gaussian-extremizable and let $\Sigma=(\Sigma_i)_{i=1}^k$ be a gaussian extremizer for $(\V, \sigma)$. For each $j \in [m]$, let 
$$
M_j=\sum_{i \in [k]} \sigma_{+}(v_i) \left(\sum_{a \in \ar_{ij}} \V(a) \cdot \Sigma_i \cdot \V(a)^T  \right) \in \RR^{\beta(w_j)\times \beta(w_j)}.
$$

We know that $\capa(\V, \sigma)>0$ by Remark \ref{gauss-ext-pos-cap} and since
$$ 
\capa(\V,\sigma)={\prod_{j \in [m]} \det(M_j)^{\sigma_{-}(w_j)}  \over \prod_{i \in [k]} \det(\Sigma_i)^{\sigma_{+}(v_i)}},$$
we conclude that each $M_j$ is a positive definite matrix. Define
$$
g(v_i)=\Sigma_i^{-{1 \over 2}}, \forall i \in [k], \text{~and~}g(w_j)=M_j^{-{1 \over 2}}, \forall j \in [m].
$$

\noindent
\textbf{Claim:} If $g:=(g(v_i), g(w_j))_{i \in [k], j \in [m]}$ then $g \in \G_{\sigma}(\V)$.

\begin{proof}[Proof of \textbf{Claim}] 
 We begin by checking that $g \cdot \V$ satisfies equation  $(\ref{geom-eq-2})$. We have that 
\begin{align*}
&\sum_{i\in [k]} \sigma_{+}(v_i) \left( \sum_{a \in \ar_{ij}} (g \cdot \V)(a) \cdot (g \cdot \V)^T(a) \right)=\\
=&\sum_{i \in [k]} \sigma_{+}(v_i) \left( \sum_{a \in \ar_{ij}} g(ha)\cdot \V(a) \cdot g(ta)^{-1}\cdot g(ta)^{-T} \cdot \V(a)^T \cdot  g(ha) \right)\\
=&\sum_{i \in [k]} \sigma_{+}(v_i) \left( \sum_{a \in \ar_{ij}} M_j^{-{1 \over 2}}\cdot \V(a) \cdot \Sigma_i^{{1 \over 2}} \cdot \Sigma_i^{{1 \over 2}} \cdot \V(a)^T \cdot  M_j^{-{1 \over 2}} \right)\\
=&M_j^{-{1 \over 2}}\cdot \left( \sum_{i \in [k]} \sigma_{+}(v_i) \left( \sum_{a \in \ar_{ij}}  \V(a) \cdot \Sigma_i \cdot \V(a)^T \right) \right) \cdot  M_j^{-{1 \over 2}} \\
=&M_j^{-{1 \over 2}}\cdot M_j \cdot  M_j^{-{1 \over 2}} =\Id_{\beta(w_j)}, \forall j \in [m],
\end{align*}
i.e. $g \cdot \V$ satisfies equation $(\ref{geom-eq-2})$. To show that $g \in \G_{\sigma}(\V)$, it remains to check that $g \cdot \V$ satisfies equation $(\ref{geom-eq-1})$, as well. For this, we first show that
\begin{equation}\label{g-extremisers-eqn}
\sum_{j \in [m]} \sigma_{-}(w_j) \left( \sum_{a \in \ar_{ij}} \V(a)^T\cdot M_j^{-1} \cdot \V(a) \right)=\Sigma_i^{-1}, \forall i \in [k].
\end{equation}
To prove it, we take logarithms in the formula for the capacity in Lemma \ref{lemma-compute-cap} and note that $\Sigma$ is a minimizer for the quantity
\begin{equation} \label{eqn-proof-g-min}
\sum_{j \in [m]}\sigma_{-}(w_j)\log \left( \det \left( \sum_{i \in [k]} \sigma_{+}(v_i) \left( \sum_{a \in \ar_{ij}} \V(a) \cdot \Sigma_i \cdot \V(a)^T \right) \right) \right)-\sum_{i \in [k]} \sigma_{+}(v_i) \log (\det(\Sigma_i))
\end{equation}
Now, let us fix an $i \in [k]$ and let $Q_i \in \RR^{\beta(v_i) \times \beta(v_i)}$ be an arbitrary symmetric matrix. Replacing $\Sigma_i$ by $\Sigma_i+\epsilon Q_i$ in $(\ref{eqn-proof-g-min})$, we can see that $\epsilon=0$ is a minimizer for 
$$
A(\epsilon)-B(\epsilon),
$$
where
\begin{align*}
&A(\epsilon)=\sum_{j \in [m]}\sigma_{-}(w_j)\log \left( \det \left( M_j+\epsilon \sigma_{+}(v_i) \left( \sum_{a \in \ar_{ij}} \V(a)\cdot  Q_i \cdot \V(a)^T \right) \right) \right),
\end{align*}
and
\begin{align*}
&B(\epsilon)=\sum_{i' \in [k], i' \neq i} \sigma_{+}(v_{i'}) \log (\det(\Sigma_{i'}))+\sigma_{+}(v_i) \log (\det(\Sigma_i+\epsilon Q_i)).
\end{align*}
Using the matrix differentiation formula ${d \over d\epsilon} \log (\det(X+\epsilon Y))|_{\epsilon=0}=\tr(X^{-1}\cdot Y)$ for any positive definite matrix $X$ and any symmetric matrix $Y$ (see for example \cite[Appendix A.2, Lemma 2.4]{Art-AviKapHai-2020}), we get that
\begin{align*}
0=&{d \over d\epsilon}(A(\epsilon)-B(\epsilon))|_{\epsilon=0}\\
=&\sum_{j \in [m]}\sigma_{-}(w_j)\tr\left(\sigma_{+}(v_i)M_j^{-1} \left( \sum_{a \in \ar_{ij}} \V(a)\cdot  Q_i \cdot \V(a)^T \right)\right)-\sigma_{+}(v_i) \tr(\Sigma_i^{-1}\cdot Q_i)
\\
=&\sigma_{+}(v_i) \left(  \sum_{j \in [m]}\tr\left(\sigma_{-}(w_j)M_j^{-1} \left( \sum_{a \in \ar_{ij}} \V(a)\cdot  Q_i \cdot \V(a)^T \right)\right)-\tr(\Sigma_i^{-1}\cdot Q_i)   \right). 
\end{align*}
Rearranging the factors inside the first trace, we obtain that  
\begin{equation} \label{tr-formula-g-min}
\tr \left( Q_i \cdot \left( \sum_{j \in [m]}\sigma_{-}(w_j) \left( \sum_{a \in \ar_{ij}} \V(a)^T\cdot M_j^{-1}\cdot \V(a) \right)-\Sigma_i^{-1} \right) \right)=0,
\end{equation}
holds for all symmetric matrices $Q_i$. As $\sum_{j \in [m]}\sigma_{-}(w_j) \left( \sum_{a \in \ar_{ij}} \V(a)^T \cdot M_j^{-1}\cdot \V(a) \right)-\Sigma_i^{-1}$ is a symmetric matrix, we can see that $(\ref{tr-formula-g-min})$ yields the desired formula $(\ref{g-extremisers-eqn})$. 

Finally, for each $i \in [k]$, we get via $(\ref{g-extremisers-eqn})$ that 
\begin{align*}
&\sum_{j \in [m]} \sigma_{-}(w_j) \left(  \sum_{a \in \ar_{ij}} (g\cdot \V)(a)^T \cdot (g\cdot V)(a) \right)=\\
=&g(v_i)^{-T}\cdot \left(\sum_{j \in [m]} \sigma_{-}(w_j) \left(  \sum_{a \in \ar_{ij}} \V(a)^T\cdot g(w_j)^T\cdot g(w_j)^ \cdot V(a) \right) \right) \cdot g(v_i)^{-1}\\
=&g(v_i)^{-T}\cdot \left(\sum_{j \in [m]} \sigma_{-}(w_j) \left(  \sum_{a \in \ar_{ij}} \V(a)^T\cdot M_j^{-1} \cdot V(a) \right) \right) \cdot g(v_i)^{-1}\\
=&g(v_i)^{-T}\cdot \Sigma_i^{-1} \cdot g(v_i)^{-1}=\Id_{\beta(v_i)},
\end{align*}
i.e. $g \cdot \V$ satisfies equation $(\ref{geom-eq-2})$, as well. This finishes the proof of our claim.
\end{proof}

It now follows from Theorem \ref{quiver-geom-data-thm} and the claim above that $\V$ is indeed $\sigma$-polystable, and the gaussian extremiser $\Sigma$ is of the form $(g(v_i)^{-1} \cdot g(v_i)^{-T})_{i \in [k]}$ with $g \in \G_{\sigma}(\V)$.
\end{proof}

As an immediate consequence of the arguments in the proof of Theorem \ref{gaussian-extremizers-thm}, we get another characterization of gaussian-extremizers.

\begin{corollary} Let $(\V, \sigma)$ is a quiver datum with $\V \in \rep(Q, \beta)$ and let 
$\mathcal{E}(\V, \sigma)$ be the possibly empty set consisting of all $k$-tuple $\Sigma=(\Sigma_i)_{i=1}^k$ with $\Sigma_i \in \s^{+}_{\beta(v_i)}$, $i \in [k]$, such that
\begin{enumerate}
\item  $M_j:=\sum_{i \in [k]} \sigma_{+}(v_i) \left(\sum_{a \in \ar_{ij}} \V(a) \cdot \Sigma_i \cdot \V(a)^T  \right) \in \RR^{\beta(w_j)\times \beta(w_j)}$ is invertible for every $j \in [m]$; and \\

\item $\Sigma_i^{-1}=\sum_{j \in [m]} \sigma_{-}(w_j) \left( \sum_{a \in \ar_{ij}} \V(a)^T\cdot M_j^{-1} \cdot \V(a) \right)$ for every $i \in [k]$.\\
\end{enumerate}
\noindent
Then $(\V, \sigma)$ is gaussian-extremizable if and only if $\mathcal{E}(\V, \sigma) \neq \emptyset$. Moreover, $\mathcal{E}(\V, \sigma)$ is the set of all gaussian-extremizers for $(\V,\sigma)$.
\end{corollary}

Finally, we give necessary and sufficient conditions for a quiver datum to have unique gaussian extremisers. 

\begin{theorem} (see also \cite[Theorem 21]{ChiDer-2021}) \label{uniqueness-g-extremals-thm} Let $(\V, \sigma)$ be a quiver datum with $0 \neq V \in \rep(Q,\beta)$. 
\begin{enumerate}
\item If $\V$ is a Schur $\sigma$-stable representation then $(\V, \sigma)$ has unique gaussian extremizers (up to scaling). 

\item If $(\V,\sigma)$ has unique gaussian extremizers (up to scaling) then $\V$ is $\sigma$-stable. 
\end{enumerate}
\end{theorem}

\begin{proof} According to Kempf-Ness theory (see for example \cite[Theorem 1.1(i)]{Boh-Laf-2017}), if $\V$ is $\sigma$-polystable then for any two $g_1, g_2 \in \G_{\sigma}(\V)$, we have that
$$
g_2 \cdot \V \in \mathbf{O}(\beta)\cdot(g_1 \cdot \V),
$$
where $\mathbf{O}(\beta)$ denotes the subgroup of $\GL(\beta)$ consisting of all tuples of orthogonal matrices. In other words,
\begin{equation} \label{orbit-g-extremals-eqn}
g_2^{-1}\cdot h \cdot g_1 \in \Stab_{\GL(\dd)}(\V)=\End_Q(\V)^{\times} \text{~for some~} h \in \mathbf{O}(\beta).
\end{equation}

\smallskip
\noindent
(1)~ Let us assume that $\V$ is a Schur $\sigma$-stable representation. We know from Theorem \ref{gaussian-extremizers-thm} that $(\V,\sigma)$ is gaussian-extremizable, and let $g_1, g_2 \in \G_{\sigma}(\V)$. Then, by $(\ref{orbit-g-extremals-eqn})$, we can write
$$
g_2=\lambda (h \cdot g_1)
$$ 
for some $h \in \mathbf{O}(\beta)$ and $\lambda \in \RR^{\times}$. Consequently, we get that
$$
g_2(v_i)^{-1} \cdot g_2(v_i)^{-T}=\lambda^{-2} (g_1(v_i)^{-1} \cdot g_1(v_i)^{-T}), \forall i \in [k].
$$
It now follows from Theorem \ref{gaussian-extremizers-thm} that $(\V, \sigma)$ has unique gaussian extremizers, up to scaling.

(2)~ Let us assume that $(V, \sigma)$ has unique gaussian-extremizers. By Theorem \ref{gaussian-extremizers-thm}, we know that $\V$ is $\sigma$-polystable. Assume for a contradiction that $\V=\V_1\oplus \V_2$ with $\V_1 \in \rep(Q,\beta_1)$, $\V_1 \in \rep(Q,\beta_2)$, two proper $\sigma$-polystable subrepresentations of $\V$. Then, by Theorem \ref{gaussian-extremizers-thm}, we know that there exist $g_1 \in \G_{\sigma}(\V_1)$ and $g_2 \in \G_{\sigma}(\V_2)$. Choosing any two scalars $\lambda_1, \lambda_2 \in \RR^{\times}$ with $|\lambda_1| \neq |\lambda_2|$, we get that the gaussian extermizers for $(\V, \sigma)$ corresponding to
$$
\left(
\begin{matrix}
g_1&0\\
0&g_2
\end{matrix}
\right)
\text{~and~}
\left(
\begin{matrix}
\lambda_1 g_1&0\\
0&\lambda_2 g_2
\end{matrix}
\right)
$$
are not a scalar multiple of each other (contradiction). This finishes the proof. 
\end{proof}

We are now ready to prove Theorem \ref{main-thm}.

\begin{proof}[Proof of Theorem \ref{main-thm}] Part (1) follows from Theorem \ref{semi-stab-cap-thm} and Lemma \ref{lemma-compute-cap} applied to the quiver datum $(V_{\A, \cc}, \sigma_{\cc, \tup})$. Parts (2), (3), and (4) follow from Theorem \ref{quiver-geom-data-thm} and Lemma \ref{lemma-compute-cap} applied to $(V_{\A, \cc}, \sigma_{\cc, \tup})$. Finally,  parts (5) and (6) follow from Theorems \ref{gaussian-extremizers-thm} and \ref{uniqueness-g-extremals-thm} applied to $(V_{\A, \cc}, \sigma_{\cc, \tup})$.
\end{proof}

\subsection*{Acknowledgment} The authors would like to thank Peter Pivovarov and Petros Valettas for many useful discussions on the paper. 

C. Chindris is supported by Simons Foundation grant $\# 711639$. H. Derksen is supported by NSF grant DMS $2147769$.

\end{document}